\documentclass[aps,pra,twocolumn]{revtex4}
\usepackage{graphicx}
\usepackage{bm}
\usepackage{dsfont,amsmath,amssymb}
\usepackage{amsthm}

\newcommand{\Tr}{\mathrm{Tr}}
\newcommand{\sgn}{\mathrm{sgn}}
\newcommand{\sinc}{\mathrm{sinc}}
\renewcommand{\Re}{\mathrm{Re}}
\renewcommand{\Im}{\mathrm{Im}}
\newtheorem*{lem}{Lemma}
\begin{document}
\title{Refined Weak Coupling Limit: Coherence, Entanglement and Non-Markovianity}
\author{\'{A}ngel Rivas}
\affiliation{Departamento de F\'isica Te\'orica I, Facultad de Ciencias F\'isicas,
Universidad Complutense, 28040 Madrid, Spain.\\
CCS -Center for Computational Simulation, Campus de Montegancedo UPM, 28660 Boadilla del Monte, Madrid, Spain.}

\date{\today}

\begin{abstract}
We study the properties of a refined weak coupling limit that preserves complete positivity in order to describe non-Markovian dynamics in the spin-boson model. With this tool, we show the system presents a rich and new non-Markovian phenomenology. This implies a dynamical difference between entanglement and coherence: the latter undergoes revivals whereas the former not, despite the induced dynamics being fully incoherent. In addition, the evolution presents ``quasieternal'' non-Markovianity, becoming non-divisible at any time period where the system evolves qualitatively. Furthermore, the method allows for an exact derivation of a master equation that accounts for a reversible energy exchange between system and environment. Specifically, this is obtained in the form of a time-dependent Lamb shift term.
\end{abstract}

%\pacs{03.65.Yz, 03.65.Aa, 42.50.Lc, 03.67.Mn}

\maketitle

\section{Introduction}
The description and characterization of non-Markovian quantum dynamics has been and is an active area of research \cite{LindbladNoMarko,Diosi,BrPe02,GardinerZoller04,Wolf,AlonsoVega,BrLaPi,RHP,Libro,FlemingHu,ChruMani2014,Review,ReviewBLPV,InesReview,Kavan}. Besides the fundamental point of view, this has been motivated by the potential utility of non-Markovian dynamics in different contexts such as quantum metrology and hypothesis testing \cite{Sun,Chin-metrology,Matsuzaki,DAA}, preservation of entanglement and coherence \cite{Huelga,Cialdi,Chen,Xu,Orieux}, and quantum information and computation \cite{Bylicka,Laine,Xiang}.

At sufficiently short times the dynamics of any open quantum system is expected to be non-Markovian \cite{RevKoss,BrPe02,Libro,GardinerZoller04,FlemingHu}. This is because the Born-Markov-secular approximation is no longer valid at a time scale smaller (or of the same order) than the width of the bath correlation functions, and so the evolution is not given by a quantum dynamical semigroup \cite{RevKoss,Alicki}. Nevertheless, provided that the system-bath coupling is small enough to justify the second order perturbation treatment, several approaches have been suggested to deal with the dynamics in the short time scale. For instance, one approach avoids the secular approximation and considers the so-called Redfield equation \cite{Redfield}. However, it has been shown this equation does not preserve positivity \cite{Dumcke} (see also \cite{Whitney,Zhao}). The schemes to overcome this last drawback range from restrictions of valid system states to the subset that remains positive under that dynamical equation \cite{Suarez}, to the inclusion of slippage operators \cite{Gaspard}. Although these proposals can be useful in some situations, they do not provide a completely general and satisfactory answer. For example, they may present problems for multipartite systems \cite{Benatti2003,Benatti2006,Benatti2007}.

Alternatively, in \cite{Schaller} Schaller and Brandes proposed a method they refer to as ``Dynamical Coarse Graining'' which, in principle, allows for a description of the second order dynamics for all time scales in a completely positive way (see also \cite{Libro}). This proposal has been successfully applied to several situations \cite{SKB,Zedler,higher,Benatti}, and it can be seen as a ``refined'' weak coupling limit \cite{Benatti}. However, as far as we know \cite{higher}, low attention has been paid to study whether or not it accounts for the non-Markovian features expected at the short time scale.

In this regard, the non-Markovian properties of the paradigmatic spin-boson model \cite{Leggett,Weiss} are still poorly understood. This model applies to a two-level system interacting linearly with a thermal bath of bosons at some temperature $T$. It plays a central role in solid state physics \cite{Chakravarty,Bulla,Chin}, chemical physics \cite{Garg,Egger,Nitzan}, quantum optics \cite{Cohen,Diego,Recati}, or quantum information technologies \cite{Shnirman,Loss,Niemczyk}. However, the absence of a completely positive description for the spin-boson model out of the Markovian regime makes the analysis in terms of measures of non-Markovianity problematic \cite{ClosBreuer2012}. These allow us to analyze in a quantitative and rigorous way to what extent the model presents non-Markovian behavior. Crucially, the positivity preservation is essential when applying measures of non-Markovianity. They are typically non-linear functions of the dynamics which have only a clear meaning under the presupposition that the dynamics is physical and preserve the positivity of the density matrix. Specifically, this implies the celebrated complete positivity condition in the case of initial system-environment factorization (see e.g. \cite{Libro}).

The objective of this paper comprises both problems by applying the refined weak coupling limit to study in detail the spin-boson model at finite temperature $T$, and examine its non-Markovian features. Particularly, we highlight the following findings:

i) We solve the dynamics of the transverse spin-boson model using the refined weak coupling method and obtain the exact Liouvillian operator for this dynamics. To our knowledge, this represents the most precise positivity-preserving master equation among the ones proposed for this problem.

ii) We find a new time-dependent Lamb shift, describing damped oscillations towards the standard Lamb shift value in the long time scale.

iii) We obtain that non-Markovianity increases for low temperatures and the system presents ``quasieternal'' non-Markovianity at $T=0$. Namely, the dynamics is non-divisible at any time instant during the period of time where the system state changes appreciably.

iv) We show a new dynamical feature between entanglement and coherence. Despite the dynamics being fully incoherent, the non-Markovian evolution may induce re-coherence cycles but does not generate entanglement revivals.

\section{Refined weak coupling limit} 
Though originally exposed in a slightly different terms, the idea behind the refined weak-coupling limit of Schaller and Brandes can be succinctly explained as follows. The exact dynamics of some open system $S$ is formally given by $\rho_S(t)=\Tr_E[U(t,t_0)\rho_S(t_0)\otimes\rho_E(t_0)U^\dagger(t,t_0)]$, with $\rho_S(t_0)$ and $\rho_E(t_0)$ the open system and environmental initial states, respectively; and $U(t,t_0)$ the unitary operator describing the joint evolution of system and environment. For some generic Hamiltonian $H=H_S+H_E+V$, with system ($H_S$) and environment ($H_E$) Hamiltonians, and interaction term $V$, the system evolution in the interaction picture and up to second order $V$ (or equivalently for short times) can be written as $(t_0=0)$
\begin{align}\label{DCG1}
  \tilde{\rho}_S(t)&=\rho_S(0)\nonumber \\
  &-\frac{1}{2}\mathcal{T}\int_0^tdt_1\int_0^tdt_2\Tr_E\left[\tilde{V}(t_1),\left[\tilde{V}(t_2),\rho_S(0)\otimes\rho_\beta\right]\right] \nonumber \\
  &+\mathcal{O}(V^3),
\end{align}
where $\tilde{X}(t)$ stands for the interaction picture of the operator $X$, $\mathcal{T}$ is the time-ordering operator, and we have already assumed the environment to be in thermal equilibrium (bath) $\rho_E(0)=\rho_\beta=\exp(-\beta H_E)/\Tr[\exp(-\beta H_E)]$. By applying $\mathcal{T}$ under the integral signs and reordering terms we obtain
\begin{align}\label{DCG2}
\tilde{\rho}_S&(t)\simeq \rho_S(0)-i[\Lambda(t),\rho_S(0)] \\
&+\Tr_E\big[W(t)\rho_S(0)\otimes\rho_\beta W(t)-\tfrac{1}{2}\big\{W^2(t),\rho_S(0)\otimes\rho_\beta\big\}\big]\nonumber,
\end{align}
with Hermitian operators $\Lambda(t)=\tfrac{1}{2i}\int_0^tdt_1\int_0^tdt_2\sgn(t_1-t_2)\Tr_E[\tilde{V}(t_1)\tilde{V}(t_2)\rho_\beta]$ and $W(t):=\int_0^t\tilde{V}(t')dt'$. Then, by writing $V=\sum_k A_k\otimes B_{k}$ with also Hermitian $A_k$ and $B_k$ and after a bit of algebra in Eq.~\eqref{DCG2} we find (for further details, see Appendix \ref{app_B})
\begin{align}\label{Z}
  \tilde{\rho}_S(t)&\simeq \rho_S(0)-i[\Lambda(t),\rho_S(0)]\nonumber \\
  &+\sum_{\omega,\omega'}\sum_{k,l} \Gamma_{kl}(\omega,\omega',t)\big[A_l(\omega')\rho_S(0) A_k^\dagger(\omega)\\
  &-\tfrac{1}{2}\{A_k^\dagger(\omega) A_l(\omega'),\rho_S(0) \}\big]\equiv\rho_S(0)+\mathcal{Z}(t)[\rho_S(0)].\nonumber 
\end{align}
where we have used the decomposition of $A_k=\sum_{\omega}A_k(\omega)$ in eigenoperators of the system Hamiltonian, $[H_S,A_k(\omega)]=-\omega A_k(\omega)$, and 
\begin{align}
\Gamma_{kl}(\omega,\omega',t)&=\textstyle\int_{0}^{t}dt_1\int_{0}^{t}dt_2 e^{i(\omega t_1-\omega' t_2)}\Tr[\tilde{B}_k(t_1-t_2)B_l\rho_\beta].
\end{align}
Similarly, the Hamiltonian correction becomes $\Lambda(t)=\sum_{\omega,\omega'}\sum_{k,l}\Xi_{kl}(\omega,\omega',t)A_k^\dagger(\omega) A_l(\omega')$ with
\begin{multline}
\Xi_{kl}(\omega,\omega',t)=\tfrac{1}{2i}\textstyle\int_{0}^{t}dt_1\int_{0}^{t}dt_2\sgn(t_1-t_2) \\
\cdot e^{i(\omega t_1-\omega' t_2)}\Tr[\tilde{B}_k(t_1-t_2)B_l\rho_\beta].
\end{multline}
From Eq. \eqref{DCG2} we infer that $\mathcal{Z}(t)$ has the GLKS form \cite{Koss-Lind}, so it turns out that the coefficients $\Gamma_{k,l}(\omega,\omega',t)$ form a positive-semidefinite matrix. Despite the fact that $\mathcal{Z}(t)$ can be seen (for fixed $t$) as the generator of a dynamical semigroup, Eq.~\eqref{Z} does not provide a completely positive (CP) dynamics as the positivity condition can be violated at order $V^3$. Nevertheless, for weak coupling (or for short times), we can safely approximate the dynamics by $e^{\mathcal{Z}(t)}$, which is indeed CP because the GKSL form of $\mathcal{Z}(t)$. Crucially, it can be proven \cite{Schaller} that for long times $\mathcal{Z}(t) \approx \mathcal{L}_{D} t$, where $\mathcal{L}_D$ is the standard generator of the weak coupling limit \cite{Davies,BrPe02,GardinerZoller04,Libro}. Thus, the refined weak coupling limit consists in taking the quantity $\mathcal{Z}(t)$, which we refer to as the \emph{Schaller-Brandes exponent}, and describing the evolution by $e^{\mathcal{Z}(t)}\rho_S(0)$. This is CP for all times, gives the exact correct dynamics at short times, and reproduces the celebrated Born-Markov-secular generator for long times.

Finally, one may ask about the Liouvillian operator $\mathcal{L}_Z(t)$ such that the solution to the differential equation $\tfrac{d \tilde{\rho}_S(t)}{dt}=\mathcal{L}_Z(t)[\tilde{\rho}_S(t)]$ gives the refined weak coupling evolution $\tilde{\rho}_S(t)=e^{\mathcal{Z}(t)}\rho_S(0)$. To this end, we write
\begin{align}
\tfrac{d\tilde{\rho}_S(t)}{dt}&=\left[\tfrac{d}{dt}e^{\mathcal{Z}(t)}\right]\rho_S(0)=\left\{\left[\tfrac{d}{dt}e^{\mathcal{Z}(t)}\right]e^{-\mathcal{Z}(t)}\right\}\tilde{\rho}_S(t)\nonumber\\
&\Rightarrow\mathcal{L}_Z=\left\{\left[\tfrac{d}{dt}e^{\mathcal{Z}(t)}\right]e^{-\mathcal{Z}(t)}\right\}.
\end{align}
Combining this with the well-known identity \cite{dExpSM} for the derivative of the exponential of an operator, $\frac{d}{dt}e^{\mathcal{Z}(t)}=\int_0^1ds e^{s\mathcal{Z}(t)} \left[\frac{d\mathcal{Z}(t)}{dt}\right] e^{(1-s)\mathcal{Z}(t)}ds$, we obtain the Liouvillian from the Schaller-Brandes exponent by means of the relation
\begin{equation}\label{LiouvillianFormula}
\mathcal{L}_Z=\int_0^1ds e^{s\mathcal{Z}(t)} \left[\frac{d\mathcal{Z}(t)}{dt}\right] e^{-s\mathcal{Z}(t)}ds.
\end{equation}
Of course, this refined weak coupling Liouvillian satisfies $\lim_{t\rightarrow\infty}\mathcal{L}_Z(t)=\mathcal{L}_D$.

%%%%%%%%%%%%%%%%%%%%%%%%%%
%%%%%%%%%%%%%%%%%%%%%%%%%%
\begin{figure*}[t]
	\includegraphics[width=\textwidth]{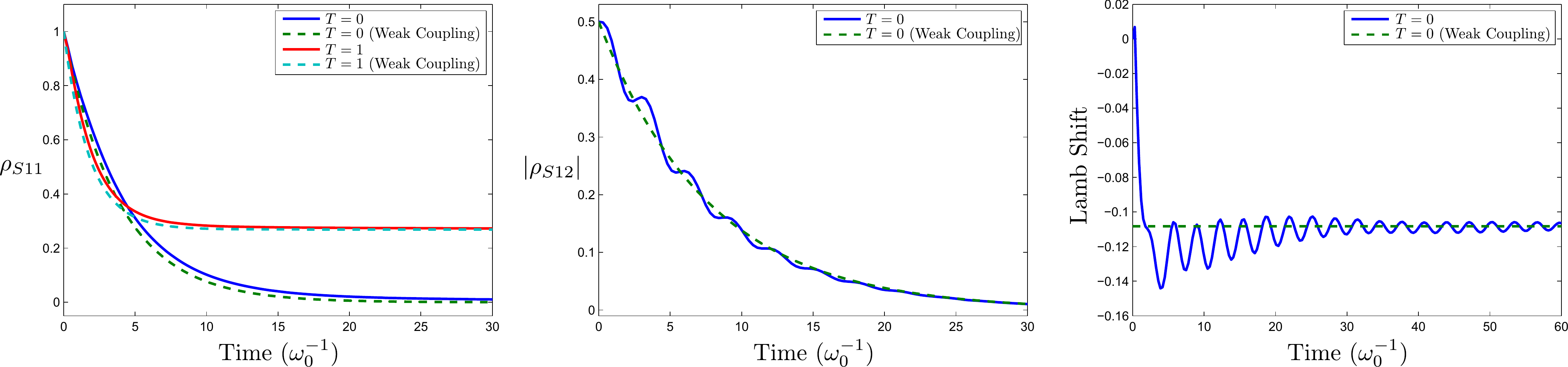}
	\caption{Dynamics in the refined weak coupling limit. The excited state population (left) shows a similar decay as for the standard weak coupling case, even in the short-time regime. The amount of coherence encoded in the off-diagonal element $\rho_{S12}$ (middle) undergoes oscillations approaching the weak coupling value in the long time scale. The Lamb shift becomes time-dependent (right); it presents long lived small oscillations towards the static weak coupling value in the asymptotic limit. For these computations we have taken an Ohmic spectral density with exponential cut-off, and parameters $\alpha=0.05$ and $\omega_c=5\omega_0$ (see main text).}
	\label{fig_1}
\end{figure*}
%%%%%%%%%%%%%%%%%%%%%%%%%%
%%%%%%%%%%%%%%%%%%%%%%%%%%

\section{Spin-boson model in the refined weak coupling limit} 

For the spin-boson model the system, environment, and interaction Hamiltonians are given by $H_S=\tfrac{\omega_0}{2}\sigma_z$, $H_E=\sum_k\omega_k a_k^\dagger a_k$, and $V=\sum_k g_k\sigma_x(a_k+a_k^\dagger)$, respectively, with Pauli matrices $\sigma_x$ and $\sigma_z$, and bosonic bath operators $a_k$. In this case the eigenoperators are $A_1(\mp\omega_0)=\sigma_\pm=(\sigma_x\pm i\sigma_y)/2$, and the computation of the Schaller-Brandes exponent for this model yields (Appendix \ref{app_B}, Sec. 3)
\begin{align}\label{Zsb}
\mathcal{Z}(t)[\rho_S]=&-i[\Xi(t,T)\sigma_z,\rho_S]\nonumber \\
&+\sum_{\mu,\nu=+,-}\Gamma_{\mu\nu}(t,T)[\sigma_\nu\rho_S\sigma_\mu^\dagger-\{\sigma_\mu^\dagger\sigma_\nu,\rho_S\}].
\end{align}
Here the coefficients are given by
\begin{widetext}
\begin{align}
&\Xi(t,T)=\frac{1}{4\pi}\int_{-\infty}^\infty d\omega t^2\left\{\mathrm{sinc}^2\left[\tfrac{(\omega_0-\omega)t}{2}\right]-\mathrm{sinc}^2\left[\tfrac{(\omega_0+\omega)t}{2}\right]\right\}\left\{\mathrm{P.V.}\int_0^\infty d\upsilon J(\upsilon)\left[\tfrac{\bar{n}_T(\upsilon)+1}{\omega-\upsilon}+\tfrac{\bar{n}_T(\upsilon)}{\omega+\upsilon}\right]\right\},\\
&\Gamma_{--}(t,T)=\int_0^\infty d\omega t^2J(\omega)\left\{[\bar{n}_T(\omega)+1]\mathrm{sinc}^2\left[\tfrac{(\omega_0-\omega)t}{2}\right]+\bar{n}_T(\omega)\mathrm{sinc}^2\left[\tfrac{(\omega_0+\omega)t}{2}\right]\right\},\\
&\Gamma_{++}(t,T)=\int_0^\infty d\omega t^2J(\omega)\left\{[\bar{n}_T(\omega)+1]\mathrm{sinc}^2\left[\tfrac{(\omega_0+\omega)t}{2}\right]+\bar{n}_T(\omega)\mathrm{sinc}^2\left[\tfrac{(\omega_0-\omega)t}{2}\right]\right\},\\
&\Gamma_{+-}(t,T)=\Gamma_{-+}^\ast(t,T)=\int_0^\infty d\omega t^2J(\omega)[2\bar{n}_T(\omega)+1]e^{-i\omega_0t}\mathrm{sinc}\left[\tfrac{(\omega_0+\omega)t}{2}\right]\mathrm{sinc}\left[\tfrac{(\omega_0-\omega)t}{2}\right],
\end{align}
\end{widetext}
where $J(\omega)$ is the spectral density of the bath, $\bar{n}_T(\omega)=[\exp(\omega/T)-1]^{-1}$ is the mean number of bosons in the bath with frequency $\omega$, and $\mathrm{sinc}(\omega):=\tfrac{\sin \omega}{\omega}$.

In Fig. 1 (left and middle) we represent the population $\rho_{S11}$ and coherence $|\rho_{S12}|$ in the refined weak coupling. We compare them with their values in the standard weak coupling (semigroup $e^{\mathcal{L}_D t}$) for different temperatures and for an Ohmic spectral density $J(\omega)=\alpha \omega e^{-\omega/\omega_c}$ $(\alpha=0.05, \omega_c=5\omega_0)$. It can be seen both dynamics differ in the small time regime where the standard weak coupling limit fails, but approach the same value for long times as expected.

Notably, it is possible to obtain a closed expression for the refined weak coupling Liouvillian for the spin-boson model. This is so because the different summands in the Schaller-Brandes exponent, Eq. \eqref{Zsb}, close a Lie algebra. This, jointly with Eq. \eqref{LiouvillianFormula}, leads to a Liouvillian with the same form as $\mathcal{Z}(t)$:
\begin{align}\label{LZ}
\frac{d\tilde{\rho}_S}{dt}&=\mathcal{L}_Z(t)[\rho_S]=-i[\Delta(t,T)\sigma_z,\rho_S]\nonumber \\
&+\sum_{\mu,\nu=+,-}\gamma_{\mu\nu}(t,T)[\sigma_\nu\rho_S\sigma_\mu^\dagger-\{\sigma_\mu^\dagger\sigma_\nu,\rho_S\}],
\end{align}
where the explicit expressions of $\Delta(t,T)$ and $\gamma_{\mu\nu}(t,T)$ are provided in Appendix \ref{app_A}. This is a very remarkable result, because, to the best of our knowledge, this is the most accurate master equation for the weakly coupled spin-boson model that guarantees complete positivity. Furthermore, it allows one to study how decay rates and energy shifts vary as a function of the time in the short time scale. For instance ,in Fig. 1 (right), we depict the evolution of the refined weak-coupling Lamb shift $\Delta(t,0)$. It is shown the way that the energy levels are initially renormalized and reach, after some transient, the standard Lamb shift as computed by the weak coupling procedure. Notably, the oscillations of $\Delta(t,0)$ decay very slow, and account for a reversible exchange of energy between system and environment not predicted with the standard semigroup approach. The experimental determination of $\Delta(t,0)$ may be used as an indicator for the time when system and environment started interacting and the strength of this interaction. 

%%%%%%%%%%%%%%%%%%%%%%%%%%
%%%%%%%%%%%%%%%%%%%%%%%%%%
\begin{figure*}[t]
	\includegraphics[width=\textwidth]{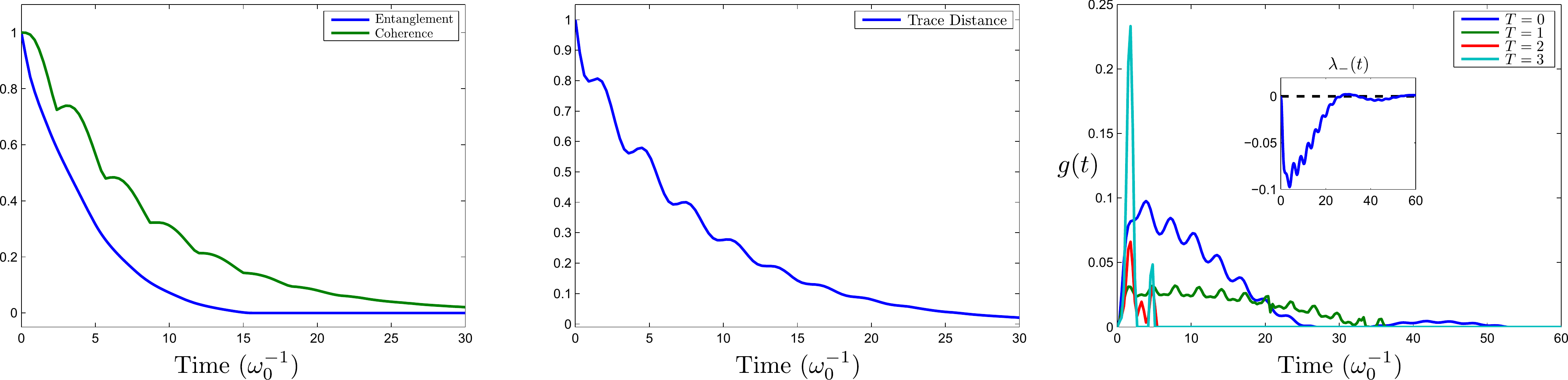}
	\caption{Non-Markovian features of the spin-boson model in the refined weak coupling limit. The entanglement (logarithmic negativity) between the system and an inert ancilla initially prepared in a maximally entangled state $|\Phi\rangle$ decays monotonically. However, the coherence ($l_1$-measure of coherence \cite{QC}) of the same state presents revivals (left). The trace distance between the two $\pm1$-eigenstates of $\sigma_y$ also shows a non-monotonic decay (middle). The $g(t)$ function \cite{RHP} is also plotted for several temperatures (right). As expected, the dynamics is non-Markovian (non-divisible) in the short-time scale. For $T=0$ the dynamics is non-divisible at any time instant in the period where the system evolves qualitatively (up to $27\omega_0^{-1}$ in the plot). The inset plot in the right part shows the time evolution of the smallest canonical decay rate. For these computations we have taken an Ohmic spectral density with exponential cut-off, and parameters $\alpha=0.05$, and $\omega_c=5\omega_0$ (see main text).}
	\label{fig_2}
\end{figure*}
%%%%%%%%%%%%%%%%%%%%%%%%%%
%%%%%%%%%%%%%%%%%%%%%%%%%%

\section{Non-Markovianity in the spin-boson model}
It is easy to check that for any diagonal state $\rho_{\rm d}$ in the basis of eigenstates of $H_S$, $\mathcal{Z}(t)[\rho_{\rm d}]$ is also a diagonal state. Thus $e^{\mathcal{Z}(t)}$ is an incoherent operation \cite{QC} and any measure of coherence must decrease monotonically for a Markovian (or CP-divisible) $e^{\mathcal{Z}(t)}$ \cite{inco}. As exemplified in Fig. 1 (middle), the absolute value of the nondiagonal component $|\rho_{S12}|$ is not monotonically decreasing, and since $|\rho_{S12}|$ is indeed a measure of coherence for a qubit \cite{QC}, we can certainly assert that the dynamics we are studying is non-Markovian.

Actually, the spin-boson model in the refined weak coupling limit presents a rich and odd non-Markovian phenomenology. A remarkable feature is that entanglement, as a difference of coherence, does present non-Markovian effects which could be used to witness non-Markovianity \cite{RHP}. More specifically, entanglement between the spin system and an inert ancilla decreases monotonically with time [see Fig. 2 (left)]. However the amount of coherence of the same state presents revivals. This strange phenomenon differentiates in a dynamical way the concepts of entanglement and coherence. In addition, the oscillatory behavior is also shared by other non-Markovinanity witnesses as the trace distance \cite{BrLaPi}, see Fig. 2 (middle). Note that in our computations we have taken the logarithmic negativity \cite{LogNeg} and the $l_1$-measure of coherence \cite{QC}, which are parallel proposals for quantifying entanglement and coherence, respectively.

The instantaneous amount of non-Markovianity can be quantified by means of the function $g(t)$ as defined in \cite{RHP}, which in terms of the Liouvillian $\mathcal{L}_Z$ has the form
\begin{equation}
g(t)=\lim_{\epsilon\rightarrow0^+}\frac{\|[\mathds{1}+\epsilon\mathcal{L}_Z(t)\otimes\mathds{1}]|\Phi\rangle\langle\Phi|\|_1-1}{\epsilon}\geq0,
\end{equation}
where $\|\cdot\|_1$ denotes the trace norm and $|\Phi\rangle$ is the maximally entangled state between the system and some ancilla of the same dimension. Alternatively, this function can also be obtained by computing the canonical decay rates of the Liouvillian \cite{MichaelHall}. One easily obtains $g(t)=\tfrac{1}{2}[|\lambda_+(t)|-\lambda_+(t)+|\lambda_-(t)|-\lambda_-(t)]$, where the canonical decay rates are
\begin{equation}
  \lambda_{\pm}(t)=\tfrac{\gamma_{++}(t)+\gamma_{--}(t)\pm \sqrt{[\gamma_{++}(t)-\gamma_{--}(t)]^2+4|\gamma_{+-}(t)|^2}}{2}.
\end{equation}
In Fig. 2 (right) we have represented the function $g(t)$ for different temperatures, obtaining a larger period of non-Markovianity at low temperatures. This behavior fit with the intuition regarding the width of the bath correlation functions, which in this case increases very rapidly as $T$ approaches zero \cite{MME}. The case of $T\rightarrow0$ is actually very relevant. In the inset plot of Fig. 2 (right) we have plotted $\lambda_{-}(t)$ for $T=0$ (bath in the vacuum). It becomes zero at long times because the refined Liouvillian $\mathcal{L}_Z(t)$ approaches to the standard weak coupling Liouvillian. This only has one nonzero decay rate at $T=0$: the one associated to the emission process related to $\lambda_+(t)$ for long times. Since the function $\lambda_{-}(t)$ remains negative for most of the time where the system evolves qualitatively, this can be thought as a form of ``quasieternal non-Markovianity''. The extreme case of ``eternal non-Markovianity'' introduced in \cite{MichaelHall} denotes the situation where the dynamics is non-Markovian for all time instant. We now see that the spin-boson model in the refined weak coupling limit presents a weak form of that case, where non-Markovianity is not kept eternally, but during the time period where the induced system change is mostly relevant.

\section{Conclusions} 
We have studied the non-Markovian features of the refined weak coupling limit proposed by Schaller and Brandes in \cite{Schaller} by applying it to the concrete example of the spin-boson model. Our conclusion is that this technique is not only able to account for highly non-Markovian effects, but that actually the spin-boson model presents a rich and new phenomenology of non-Markovianity. The amount of entanglement with an ancilla does not show revivals \cite{RHP}; however, the amount of coherence does. This surprising effect represents a new dynamical difference between entanglement and coherence. In addition, the system is more non-Markovian as the temperature decreases, and becomes non-Markovian for every time instant during the period of qualitative evolution for an environment at zero temperature. This effect recalls the case ``ethernal non-Markovianity'' theroretically proposed in \cite{MichaelHall}. We may see now that the ubiquitous spin-boson model can behave very similarly.

Furthermore, we find the time-dependent Lamb shift term, describing a reversible exchange of energy between system and environment. This is not detected by the standard weak coupling treatment where the environment leads only to irreversibilities.   

Besides the fundamental interest on these new effects, the large amount of controlled systems well described by the paradigmatic spin-boson model provides this work with a practical perspective. Thus, the results here reported are very suited to be verified experimentally in platforms of AMO and solid-state physics. 

\acknowledgments The author acknowledges to D. G\'omez for discussions on this topic. Partial financial support from Spanish MINECO grants FIS2015-67411, FIS2012-33152, a ``Juan de la Cierva-Incorporaci\'on" research contract, the CAM research consortium QUITEMAD+ S2013/ICE-2801 and the U.S. Army Research Office through grant W911NF-14-1-0103 is acknowledged.

\textit{Note Added}.--- While finalizing this work, L. Ferialdi obtained a different master equation for the spin-boson model based on an exact treatment \cite{Ferialdi}.

\appendix

\onecolumngrid

\section{Refined weak coupling Liouvillian for the spin-boson model} 
\label{app_A}
The Schaller-Brandes exponent $\mathcal{Z}(t)$ is a linear combination of the operators $\mathcal{Z}_z=[\sigma_z,\cdot]$,  $\mathcal{Z}_{+-}=\sigma_+(\cdot)\sigma_--\{\sigma_-\sigma_+,\cdot\}/2$, $\mathcal{Z}_{-+}=\sigma_-(\cdot)\sigma_+-\{\sigma_+\sigma_-,\cdot\}/2$, $\mathcal{Z}_{--} =\sigma_-(\cdot)\sigma_-$, and $\mathcal{Z}_{++}=\sigma_+(\cdot)\sigma_+$. These operators close a Lie algebra
\begin{align}
&[\mathcal{Z}_{+-},\mathcal{Z}_{-+} ]=\mathcal{Z}_{+-}-\mathcal{Z}_{-+},\quad [\mathcal{Z}_{++},\mathcal{Z}_{--} ]=\mathcal{Z}_z/2,\quad [\mathcal{Z}_z,\mathcal{Z}_{--} ]=4\mathcal{Z}_{++},\quad [\mathcal{Z}_z,\mathcal{Z}_{++} ]=-4\mathcal{Z}_{--},
\end{align}
with zero value for the rest of the cases. This, because of Eq. \eqref{LiouvillianFormula}, leads immediately to Eq. \eqref{LZ}. After a rather tedious but straightforward algebra the coefficients in the Liouvillian can be computed to be:
\begin{align}
\gamma_{++}&=\frac{1}{(\Gamma_{++}+\Gamma_{--})^2}\left\{\left[(e^{-(\Gamma_{++}+\Gamma_{--})}-1\right](\Gamma_{++}\dot{\Gamma}_{--}-\dot{\Gamma}_{++}\Gamma_{--})+(\dot{\Gamma}_{++}+\dot{\Gamma}_{--})(\Gamma^2_{++}+\Gamma_{++}\Gamma_{--})\right\},\\
\gamma_{--}&=\frac{1}{(\Gamma_{++}+\Gamma_{--})^2}\left\{\left[e^{-(\Gamma_{++}+\Gamma_{--})}-1\right](\dot{\Gamma}_{++}\Gamma_{--}-\Gamma_{++}\dot{\Gamma}_{--})+(\dot{\Gamma}_{++}+\dot{\Gamma}_{--})(\Gamma^2_{--}+\Gamma_{++}\Gamma_{--})\right\},
\end{align}
 \begin{multline}
\gamma_{+-}=\gamma_{-+}^\ast=\frac{1}{2 \left(|\Gamma_{+-}|^2-\Xi^2\right)}\left\{2\Gamma_{+-}[\Re(\dot{\Gamma}_{-+}\Gamma_{+-})-\dot{\Xi}\Xi]-i(\dot{\Gamma}_{+-}\Xi-\Gamma_{+-}\dot{\Xi})\left[1-\cosh(2\sqrt{|\Gamma_{+-}|^2-\Xi^2})\right]\right.\\
\left.+i\frac{\Gamma_{+-}\Im(\Gamma_{-+}\dot{\Gamma}_{+-})+\Xi(\Gamma_{+-}\dot{\Xi}-\dot{\Gamma}_{+-}\Xi)}{\sqrt{|\Gamma_{+-}|^2-\Xi^2}}\sinh(2\sqrt{|\Gamma_{+-}|^2-\Xi^2})\right\},
\end{multline}
 \begin{multline}
\Delta=\frac{1}{2 \left(|\Gamma_{+-}|^2-\Xi^2\right)}\left\{2\Xi[\Re(\dot{\Gamma}_{-+}\Gamma_{+-})-\dot{\Xi}\Xi]+\Im(\dot{\Gamma}_{-+}\Gamma_{+-})\left[1+\cosh(2\sqrt{|\Gamma_{+-}|^2-\Xi^2})\right]\right.\\
\left.+\frac{\Re(\dot{\Gamma}_{-+}\Gamma_{+-})\Xi-|\Gamma_{+-}|^2\dot{\Xi}}{\sqrt{|\Gamma_{+-}|^2-\Xi^2}}\sinh(2\sqrt{|\Gamma_{+-}|^2-\Xi^2})\right\}.
\end{multline}
Here, for the sake of compactness, we have omitted the $(t,T)$ dependence of the coefficients and denoted $dX/dt\equiv\dot{X}$. 

%%%%%%%%%%
%%%%%%%%%%
\section{Further Details About The Refined Weak Coupling Limit}
\label{app_B}
%%%%%%%%%
%%%%%%%%%

Consider the total Hamiltonian $H=H_S+H_E+V$ with the usual product initial condition $\rho(0)=\rho_S(0)\otimes\rho_E$, where $\rho_E$ is a stationary state of the environment $[H_E,\rho_E]=0$. In the interaction picture the reduced state at time $t$ is
\begin{equation}\label{coarsegrained1}
\tilde{\rho}_S(t)=\Tr_E\left[U(t,0)\rho_S(0)\otimes\rho_E U^\dagger(t,0)\right],
\end{equation}
where
\begin{equation}
  U(t,0)=\mathcal{T}e^{-i\int_0^t\tilde{V}(t')dt'}
\end{equation}
is the unitary propagator, $\tilde{X}(t)$ stands for the operator $X$ in the interaction picture, and $\mathcal{T}$ denotes the time-ordering operator. The propagator to the first non-trivial order in Eq. \eqref{coarsegrained1} gives
\begin{equation}
  \tilde{\rho}_S(t)=\rho_S(0)-\frac{1}{2}\mathcal{T}\int_0^tdt_1\int_0^tdt_2\Tr_E\left[\tilde{V}(t_1),\left[\tilde{V}(t_2),\rho_S(0)\otimes\rho_E\right]\right]+\mathcal{O}(V^3).
\end{equation}
Here we have made the common assumption that the first order term vanishes $\Tr[\tilde{V}(t)\rho_B]=0$ \cite{Libro}. From the definition of the time-ordering operation we obtain
\begin{align}\label{Z1Z2Z3}
\mathcal{T}\int_0^tdt_1\int_0^tdt_2\Tr_E\left[\tilde{V}(t_1),\left[\tilde{V}(t_2),\rho_S(0)\otimes\rho_E\right]\right]&=\int_0^tdt_1\int_0^tdt_2\theta(t_1-t_2)\Tr_E\left[\tilde{V}(t_1),\left[\tilde{V}(t_2),\rho_S(0)\otimes\rho_E\right]\right]\nonumber\\
&+\int_0^tdt_1\int_0^tdt_2\theta(t_2-t_1)\Tr_E\left[\tilde{V}(t_2),\left[\tilde{V}(t_1),\rho_S(0)\otimes\rho_E\right]\right]\nonumber\\
&\equiv2\mathcal{Z}_1[\rho_S(0)]+\mathcal{Z}_2[\rho_S(0)]+\mathcal{Z}_3[\rho_S(0)],
\end{align}
where, after expanding the double commutators, we find three kind of terms, $\mathcal{Z}_1$, $\mathcal{Z}_2$, and $\mathcal{Z}_3$. The first one is given by
\begin{align}
\mathcal{Z}_1[\rho_S(0)]&=-\int_0^tdt_1\int_0^tdt_2\theta(t_1-t_2)\Tr_E\left[\tilde{V}(t_1)\rho_S(0)\otimes\rho_E\tilde{V}(t_2)\right]+\theta(t_2-t_1)\Tr_E\left[\tilde{V}(t_1)\rho_S(0)\otimes\rho_E\tilde{V}(t_2)\right]\nonumber \\
%&=-\int_0^tdt_1\int_0^tdt_2[\theta(t_1-t_2)+\theta(t_2-t_1)]\Tr_E\left[\tilde{V}(t_1)\rho_S(0)\otimes\rho_E\tilde{V}(t_2)\right]\nonumber \\
&=-\int_0^tdt_1\int_0^tdt_2\Tr_E\left[\tilde{V}(t_1)\rho_S(0)\otimes\rho_E\tilde{V}(t_2)\right];
\end{align}
defining $W(t):=\int_0^t\tilde{V}(t')dt'$ we have
\begin{equation}
  \mathcal{Z}_1[\rho_S(0)]=-\Tr_E\left[W(t)\rho_S(0)\otimes\rho_EW(t)\right].
\end{equation}
The factor 2 in front of $\mathcal{Z}_1$ in Eq. \eqref{Z1Z2Z3} comes from another analogous term corresponding to the exchange $t_1\leftrightarrow t_2$ inside the double commutator.
The second term goes like
\begin{align}
\mathcal{Z}_2[\rho_S(0)]&=\int_0^tdt_1\int_0^tdt_2\theta(t_1-t_2)\Tr_E\left[\tilde{V}(t_1)\tilde{V}(t_2)\rho_S(0)\otimes\rho_E\right]+\theta(t_2-t_1)\Tr_E\left[\tilde{V}(t_2)\tilde{V}(t_1)\rho_S(0)\otimes\rho_E\right]\nonumber \\
&=\int_0^tdt_1\int_0^tdt_2[\theta(t_1-t_2)+\theta(t_2-t_1)]\Tr_E\left[\tilde{V}(t_2)\tilde{V}(t_1)\rho_S(0)\otimes\rho_E\right]\nonumber \\
&+\int_0^tdt_1\int_0^tdt_2\theta(t_1-t_2)\Tr_E\left\{[\tilde{V}(t_1),\tilde{V}(t_2)]\rho_S(0)\otimes\rho_E\right\}\nonumber\\
&=\Tr_E\left[W^2(t)\rho_S(0)\otimes\rho_E\right]+\int_0^tdt_1\int_0^tdt_2\theta(t_1-t_2)\Tr_E\left\{[\tilde{V}(t_1),\tilde{V}(t_2)]\rho_S(0)\otimes\rho_E\right\}.
\end{align}
Similarly the remaining term can be expressed as
\begin{equation}
\mathcal{Z}_3[\rho_S(0)]=\Tr_E\left[\rho_S(0)\otimes\rho_E W^2(t)\right]-\int_0^tdt_1\int_0^tdt_2\theta(t_1-t_2)\Tr_E\left\{\rho_S(0)\otimes\rho_E[\tilde{V}(t_1),\tilde{V}(t_2)]\right\}.
\end{equation}
Thus, everything together gives $\tilde{\rho}_S(t)\equiv\rho_S(0)+\mathcal{Z}(t)[\rho_S(0)]+\mathcal{O}(V^3)$ with
\begin{align}\label{ZSM1}
\mathcal{Z}(t)[\rho_S(0)]&=-\frac{1}{2}\left(2\mathcal{Z}_1[\rho_S(0)]+\mathcal{Z}_2[\rho_S(0)]+\mathcal{Z}_3[\rho_S(0)]\right)\nonumber\\
&=-i[\Lambda(t),\rho_S(0)]+\Tr_E\left[W(t)\rho_S(0)\otimes\rho_E W(t)-\frac{1}{2}\left\{W^2(t),\rho_S(0)\otimes\rho_E\right\}\right].
\end{align}
Here, the self-adjoint operator $\Lambda(t)$ is given by
\begin{align}\label{Lambda1}
\Lambda(t)&=\frac{1}{2i}\int_0^tdt_1\int_0^tdt_2\theta(t_1-t_2)\Tr_E\left\{[\tilde{V}(t_1),\tilde{V}(t_2)]\rho_E\right\} \nonumber \\
&=\frac{1}{2i}\int_0^tdt_1\int_0^tdt_2\sgn(t_1-t_2)\Tr_E\left[\tilde{V}(t_1)\tilde{V}(t_2)\rho_E\right],
\end{align}
where we have used the relation $\theta(x)=[1+\sgn(x)]/2$. By taking the spectral decomposition of $\rho_E$ we immediately check that for any fixed $t$, $\mathcal{Z}(t)$ has the GKSL form \cite{Koss-Lind}.

Since at first non-trivial order we have
\begin{equation} \label{coarsegrained2}
\tilde{\rho}_S(t)=[\mathds{1}+\mathcal{Z}(t)]\rho_A(0)+\mathcal{O}(V^3)\simeq e^{\mathcal{Z}(t)}\rho_A(0),
\end{equation}
the refined weak coupling dynamics given by $e^{\mathcal{Z}(t)}$ is a completely positive dynamical map that approaches the exact one at the short time scale. Furthermore, Schaller and Brandes \cite{Schaller} proved that for large times $e^{\mathcal{Z}(t)}$ provides a consistent second order approximation that becomes closer to the usual weak coupling dynamics. For the sake of completeness we shall reproduce their result in the following subsections.

%%%%
\subsection{Schaller-Brandes Exponent in Terms of the $H_S$ Eigenoperators}
%%%%
The interaction Hamiltonian can always be written in the form
\begin{equation}\label{VAB}
V=\sum_{k}A_k\otimes B_k,
\end{equation}
where $A_k^\dagger =A_k$, $B_k^\dagger=B_k$ are seft-adjoint operators of system and environment, respectively \cite{Libro}. Now, assuming for the sake of simplicity that the system Hamiltonian $H_S$ has discrete spectra and $|\epsilon\rangle$ are the associated eigenstates $H_S|\epsilon\rangle=\epsilon|\epsilon\rangle$, we define
\begin{equation}
A_k(\omega)=\sum_{\epsilon-\epsilon'=\omega}|\epsilon\rangle\langle\epsilon|A_k|\epsilon'\rangle\langle\epsilon'|,
\end{equation}
where the summation runs over every pair of energies $\epsilon$ and $\epsilon'$ such that their difference is $\omega$. The operators $A_k(\omega)$ so defined are in fact eigenoperators of the superoperator $[H_S,\cdot]$ with eigenvalue $-\omega$:
\begin{equation}
[H_S,A_k(\omega)]=-\omega A_k(\omega),
\end{equation}
so that in the interaction picture $\tilde{A}_k(\omega,t)=e^{-i\omega t}A_k(\omega)$. Moreover, they satisfy the following properties (see, e.g. \cite{Libro}):
\begin{align}
&A_k^\dagger(\omega)=A_k(-\omega),\\
&\sum_{\omega}A_k(\omega)=\sum_{\omega}A_k^\dagger(\omega)=A_k,\label{sumAk}\\
&[H_S,A_k^\dagger(\omega)A_l(\omega)]=0.
\end{align}
Thus the interaction Hamiltonian in the interaction picture can be written as
\begin{equation}\label{VAomega}
\tilde{V}(t)=\sum_{\omega,k}e^{-i\omega t}A_k(\omega)\otimes \tilde{B}_k(t)=\sum_{\omega,k}e^{i\omega t}A_k^\dagger(\omega)\otimes \tilde{B}_k(t).
\end{equation}
Using these decompositions in Eq. \eqref{Lambda1} we obtain
\begin{equation}
  \Lambda(t)=\sum_{\omega,\omega'}\sum_{k,l}\Xi_{kl}(\omega,\omega',t)A_{k}^{\dagger}(\omega)A_{l}(\omega'),
\end{equation}
with
\begin{align}\label{XiSM1}
\Xi_{kl}(\omega,\omega',t)=\frac{1}{2i}\int_{0}^{t}dt_1\int_{0}^{t}dt_2\sgn(t_1-t_2)e^{i(\omega t_1-\omega' t_2)}\Tr\left[\tilde{B}_k(t_1-t_2)B_l\rho_E\right].
\end{align}
Here, we have assumed that the environment is in a stationary state $[H_E,\rho_E]=0$ so that the environmental correlation functions just depend on the time difference $(t_1-t_2)$. Similarly, we can write the non-Hamiltonian part of \eqref{ZSM1} in terms of eigenoperators $A_k(\omega)$ so that Schaller-Brandes exponent yields
\begin{align}
  \mathcal{Z}(t)[\rho_S(0)]=-i[\Lambda(t),\rho_S(0)]+\sum_{\omega,\omega'}\sum_{k,l} \Gamma_{kl}(\omega,\omega',t)\big[A_l(\omega')\rho_S(0) A_k^\dagger(\omega)-\tfrac{1}{2}\{A_k^\dagger(\omega) A_l(\omega'),\rho_S(0) \}\big],
\end{align}
with:
\begin{equation}\label{GammaSM1}
\Gamma_{kl}(\omega,\omega',t)=\int_{0}^{t}dt_1\int_{0}^{t}dt_2 e^{i(\omega t_1-\omega' t_2)}\Tr\left[\tilde{B}_k(t_1-t_2)B_l\rho_E\right].
\end{equation}

%%%%
\subsection{Long time limit}
%%%%
In order to study the behavior of $\mathcal{Z}(t)$ for long times we first prove a preliminary lemma.
\begin{lem} The following identity holds true in the distributional sense:
\begin{equation}
\lim_{t\rightarrow\infty}t\ \sinc\left[\tfrac{(\omega+a)t}{2}\right]\sinc\left[\tfrac{(\omega+b)t}{2}\right]=2\pi\delta_{a,b}\delta(\omega+a).
\end{equation}
Namely, for any (sufficiently well-behaved) test function $f(\omega)$ we have
\begin{equation}\label{distribution0}
   \lim_{t\rightarrow\infty} \int_I f(\omega) t\ \sinc\left[\tfrac{(\omega+a)t}{2}\right]\sinc\left[\tfrac{(\omega+b)t}{2}\right]d\omega=2\pi\delta_{a,b}f(-a),
 \end{equation}
for $-a\in I$, and zero otherwise.
\end{lem}

\begin{proof} Let $f(\omega)$ be a differentiable function with compact support $I=(-\omega_0,\omega_0)$. Suppose $a\neq b$, using that $\sinc(x)=\sin(x)/x$ and decomposing in partial fractions we obtain
%\begin{widetext}
\begin{multline}\label{distribution1}
\lim_{t\rightarrow\infty} \int_I f(\omega) t\ \sinc\left[\tfrac{(\omega+a)t}{2}\right]\sinc\left[\tfrac{(\omega+b)t}{2}\right]d\omega= \lim_{t\rightarrow\infty}\frac{4}{(b-a)} \int_I f(\omega)  \left\{\frac{ \sin[\tfrac{(\omega+a)t}{2}]\sin[\tfrac{(\omega+b)t}{2}]}{t(\omega+a)}\right.\\
- \left.\frac{ \sin[\tfrac{(\omega+a)t}{2}]\sin[\tfrac{(\omega+b)t}{2}]}{t(\omega+b)}\right\}d\omega.
\end{multline}
Since $\big|\tfrac{\sin(x/2)}{x}\big|\leq\tfrac{1}{2}$ each of both integrands on the right hand side are dominated by $|f(\omega)|/2$, which is integrable in $I$. Then Lebesgue's dominated convergence theorem \cite{BoccaraSM} allows us to exchange the limit and the integral sign obtaining straightforwardly zero integrals.

Consider now the case $a=b$. Then, we have
\begin{equation} \label{distribution2}
    \lim_{t\rightarrow\infty} \int_I f(\omega) t\ \sinc^2\left[\frac{(\omega+a)t}{2}\right]d\omega=\lim_{t\rightarrow\infty} \left\{\int_I [f(\omega)-f(-a)] t\ \sinc^2\left[\frac{(\omega+a)t}{2}\right]d\omega+f(-a)\int_I t\ \sinc^2\left[\frac{(\omega+a)t}{2}\right]d\omega\right\},
\end{equation}
where we have added and subtracted $f(-a)$. Integrating by parts the last integral of the right hand side yields
\begin{equation}
     \lim_{t\rightarrow\infty}\int_I t\ \sinc^2\left[\frac{(\omega+a)t}{2}\right]d\omega=2\lim_{t\rightarrow\infty}\left\{\left.\frac{\cos[t (\omega+a)]-1}{t (\omega+a)}\right|_{-\omega_0}^{\omega_0}+\left.\int_0^{t(\omega+a)}\sinc(x)dx\right|_{-\omega_0}^{\omega_0}\right\}=2\pi,
\end{equation}
for $-a\in I$, as $\int_0^\infty\sinc (x) dx=\tfrac{\pi}{2}$. It is also straightforwardly checked that if $-a\notin I$, the integral vanishes. Finally, the remaining integral in Eq. \eqref{distribution2} is
\begin{equation}
\lim_{t\rightarrow\infty} \int_I [f(\omega)-f(-a)] t\ \sinc^2\left[\frac{(\omega+a)t}{2}\right]d\omega=\lim_{t\rightarrow\infty}4\int_I \frac{[f(\omega)-f(-a)]}{(\omega+a)} \frac{\sin^2\left[\frac{(\omega+a)t}{2}\right]}{(\omega+a)t}d\omega.
\end{equation}
The above integrand is dominated by the function $\left|\frac{[f(\omega)-f(-a)]}{(\omega+a)}\right|$, which has no problem in $\omega=-a$ because $f(\omega)$ is supposed to be differentiable everywhere. Therefore, the exchange of the limit and the integral sign gives the zero value.
%\end{widetext}

All of this is equally applicable to a sufficiently fast decaying function $f(\omega)$ but not necessarily with compact support. For that case we may split the integration interval in three subintervals $I=(-\infty,-\omega_0)\cup(-\omega_0,\omega_0)\cup(\omega_0,\infty)$ with $-a\in(-\omega_0,\omega_0)$. The integrals on $(-\infty,-\omega_0)$ and $(\omega_0,\infty)$ become zero due to Lebesgue's dominated convergence theorem for sufficiently fast decaying $f(\omega)$.
\end{proof}

Let us now reproduce the Schaller and Brandes result \cite{Schaller} for the long time limit of $\mathcal{Z}(t)$. Consider $\Gamma_{kl}(\omega,\omega',t)$ in Eq. \eqref{GammaSM1},
\begin{equation}\label{GammaSM2}
\Gamma_{kl}(\omega,\omega',t)=\int d\upsilon \int_{0}^{t}dt_1\int_{0}^{t}dt_2 e^{i[(\omega-\upsilon) t_1-(\omega'-\upsilon) t_2]} \Tr\left[B_k(\upsilon)B_l\rho_E\right],
\end{equation}
where we have used the decomposition $B_k=\int d\upsilon B_k(\upsilon)$ in terms of eigenoperators $B_k(\upsilon)$ of $H_E$ with frequency $\upsilon$. This is similar to Eq. \eqref{sumAk} but the sum is here substituted by an integral since the environment is assumed to have an infinite (potentially continuous) number of degrees of freedom. Performing the integrals $\int_{0}^{t}ds e^{i x s}=t e^{ixt/2}\sinc(xt/2)$ we obtain
\begin{equation}
\Gamma_{kl}(\omega,\omega',t)=\int d\upsilon t^2 \exp\left[i\tfrac{(\omega-\omega')t}{2}\right]\sinc\left[\tfrac{(\omega-\upsilon)t}{2}\right]\sinc\left[\tfrac{(\omega'-\upsilon)t}{2}\right]\Tr\left[B_k(\upsilon)B_l\rho_E\right].
\end{equation}
Therefore, due to the Lemma above, we can assert that
\begin{equation}
  \lim_{t\rightarrow\infty}\frac{\Gamma_{kl}(\omega,\omega',t)}{t}=2\pi\delta_{\omega,\omega'}\Tr\left[B_k(\upsilon)B_l\rho_E\right].
\end{equation}
The quantity $\gamma_{kl}:=2\pi\Tr\left[B_k(\upsilon)B_l\rho_E\right]$ is just the decay rate in the standard weak coupling limit \cite{Libro} and $ \delta_{\omega,\omega'}$ performs the secular approximation.

For the Hamiltonian part one needs a bit more effort. Firstly, we introduce the decomposition $B_k=\int d\upsilon B_k(\upsilon)$ in the expression $\sgn(t_1-t_2)\Tr\left[\tilde{B}_k(t_1-t_2)B_l\rho_E\right]$, obtaining:
\begin{equation}
\int d\upsilon\ \sgn(\tau)e^{-i\upsilon\tau}\left[B_k(\upsilon)B_l\rho_E\right],
\end{equation}
with $\tau=t_1-t_2$. Now we take Fourier transform with respect to $\tau$,
\begin{equation}
\int d\upsilon \int_{-\infty}^{\infty} d\tau \ \sgn(\tau)e^{i(\varphi-\upsilon)\tau}\left[B_k(\upsilon)B_l\rho_E\right].
\end{equation}
A well-known result in distribution theory says that the Fourier transform of the sign function $\sgn(\tau)$ in the distributional sense is $2i$ times the Cauchy principal value distribution \cite{BoccaraSM}, namely
\begin{equation}
2i{\rm P.V.}\int d\upsilon \frac{\left[B_k(\upsilon)B_l\rho_E\right]}{\varphi-\upsilon}.
\end{equation}
Therefore, by taking inverse Fourier transform, we find the relation
\begin{equation}\label{auxSM1}
\sgn(\tau)\Tr\left[\tilde{B}_k(\tau)B_l\rho_E\right]
=\frac{i}{\pi}\int_{-\infty}^{\infty} d\varphi e^{-i\varphi\tau} {\rm P.V.}\int d\upsilon  \frac{\left[B_k(\upsilon)B_l\rho_E\right]}{\varphi-\upsilon}.
\end{equation}
This equality combined with Eq. \eqref{XiSM1} yields
\begin{align}
\Xi_{kl}(\omega,\omega',t)=\frac{1}{2\pi}\int_{-\infty}^{\infty} d\varphi \int_{0}^{t}dt_1\int_{0}^{t}dt_2e^{i[(\omega-\varphi) t_1-(\omega'-\varphi) t_2]}  {\rm P.V.}\int d\upsilon  \frac{\left[B_k(\upsilon)B_l\rho_E\right]}{\varphi-\upsilon}.
\end{align}
Finally, by following the same steps as in Eq. \eqref{GammaSM2} for $\Gamma_{kl}(\omega,\omega',t)$, it is straightforward to prove that
\begin{equation}
\lim_{t\rightarrow\infty}\frac{\Xi_{kl}(\omega,\omega',t)}{t}=\delta_{\omega,\omega'}{\rm P.V.}\int d\upsilon  \frac{\left[B_k(\upsilon)B_l\rho_E\right]}{\omega-\upsilon}.
\end{equation}
Here $S_{kl}(\omega):={\rm P.V.}\int d\upsilon  \tfrac{\left[B_k(\upsilon)B_l\rho_E\right]}{\omega-\upsilon}$ are the shifts obtained in the standard weak coupling limit \cite{Libro} and $ \delta_{\omega,\omega'}$ performs the secular approximation, as commented.

Summarizing, we have obtained that $\lim_{t\rightarrow\infty}\mathcal{Z}(t)/t=\mathcal{L}_D$ where $\mathcal{L}_D$ is the Liouvillian of the standard weak coupling limit under the Born-Markov-secular approximation. Hence for long times both dynamics are the same $e^{\mathcal{Z}(t)}\simeq e^{\mathcal{L}_Dt}$.

%%%%%%%%%%
%%%%%%%%%%
\subsection{Schaller-Brandes Exponent for the Spin-Boson Model}
\label{app_B.1}
%%%%%%%%%
%%%%%%%%%
The interaction Hamiltonian reads $V=\sum_l A_l\otimes B_l=A_1\otimes B_1$, for $A_1=\sigma_x$ and $B_1=\sum_k g_k (a_k+a_k^\dagger)$. Moreover we have that
\begin{equation}
A_1=A_1(-\omega_0)+A_1(\omega_0), \text{ and }
B_1=\sum_k B_1(\omega_k)+B_1(-\omega_k),
\end{equation}
with $A_1(\mp\omega_0)=\sigma_\pm$, and $B_1(\omega_k)=g_k a_k$ and $B_1(-\omega_k)=g_k a_k^\dagger$, the eigenoperators of $[H_S,\cdot]$ and $[H_E,\cdot]$, respectively.

Considering the environmental modes to be in thermal equilibrium, $\rho_E=\rho_\beta=e^{-\beta H_S}/\Tr(e^{-\beta H_S})$, the bath correlation functions become
%\begin{widetext}
\begin{align}\label{CorrSM}
\Tr\left[\tilde{B}_1(t_1-t_2)B_1\rho_E\right]&=
\sum_k g_k^2 \left\{ e^{-i\omega_k(t_1-t_2)}\Tr\left[a_ka_k^\dagger\rho_\beta\right] + e^{i\omega_k(t_1-t_2)}\Tr\left[a_k^\dagger a_k \rho_\beta\right]\right\}\nonumber\\
&=\int _0^\infty d\omega J(\omega) \{e^{-i\omega (t_1-t_2)}[n_T(\omega)+1]+e^{i\omega (t_1-t_2)}n_T(\omega)\},
\end{align}
%\end{widetext}
where $\bar{n}_T(\omega)=\Tr(a_k^\dagger a_k\rho_\beta)=[\exp(\omega/T)-1]^{-1}$, and we have taken the continuous limit in the environmental modes by introducing the bath spectral density $J(\omega)\sim\sum_{k}g_k^2 \delta(\omega-\omega_k)$.

Thus, the coefficient $\Gamma(\omega,\omega',t)$ for $\omega=\omega'=-\omega_0$ in \eqref{GammaSM1} is
%\begin{widetext}
\begin{align}
\Gamma(-\omega_0,-\omega_0,t)\equiv \Gamma_{++}(t,T)&=\int_0^\infty d\omega J(\omega)\int_0^t dt_1\int_0^t dt_2 e^{-i\omega_0 (t_1-t_2)} \{e^{-i\omega (t_1-t_2)}[n_T(\omega)+1]+e^{i\omega (t_1-t_2)}n_T(\omega)\}\nonumber\\
&=\int_0^\infty d\omega t^2J(\omega)\left\{[\bar{n}_T(\omega)+1]\mathrm{sinc}^2\left[\tfrac{(\omega_0+\omega)t}{2}\right]+\bar{n}_T(\omega)\mathrm{sinc}^2\left[\tfrac{(\omega_0-\omega)t}{2}\right]\right\}.
\end{align}
%\end{widetext}
Similarly, the remaining coefficients written in the main text are $\Gamma_{--}(t,T)=\Gamma(\omega_0,\omega_0,t)$ and $\Gamma_{+-}(t,T)=\Gamma_{-+}^\ast(t,T)=\Gamma(-\omega_0,\omega_0,t)$.

For the shifts, we introduce \eqref{CorrSM} in Eq. \eqref{XiSM1} and use the Fourier transform to substitute the function $\sgn(t_1-t_2)$ in terms of the principal value of the integral as in Eq. \eqref{auxSM1}. Then by taking into account that $\sigma_+\sigma_-=(1+\sigma_z)/2, \sigma_-\sigma_+=(1-\sigma_z)/2$, and $\sigma_+^2=\sigma_-^2=0$, the Hermitian part of the Schaller-Brandes exponent can be written as
\begin{equation}
\Lambda(t)=\left[\frac{\Xi(\omega_0,\omega_0,t)-\Xi(-\omega_0,-\omega_0,t)}{2}\right]\sigma_z\equiv\Xi(t,T)\sigma_z,
\end{equation}
with
%\begin{widetext}
\begin{align}
&\Xi(\omega_0,\omega_0,t)=\frac{1}{2\pi}\int_{-\infty}^\infty d\omega t^2\sinc^2\left[\tfrac{(\omega_0-\omega)t}{2}\right]\left\{\mathrm{P.V.}\int_0^\infty d\upsilon J(\upsilon)\left[\frac{\bar{n}_T(\upsilon)+1}{\omega-\upsilon}+\frac{\bar{n}_T(\upsilon)}{\omega+\upsilon}\right]\right\},\\
&\Xi(-\omega_0,-\omega_0,t)=\frac{1}{2\pi}\int_{-\infty}^\infty d\omega t^2\sinc^2\left[\tfrac{(\omega_0+\omega)t}{2}\right]\left\{\mathrm{P.V.}\int_0^\infty d\upsilon J(\upsilon)\left[\frac{\bar{n}_T(\upsilon)+1}{\omega-\upsilon}+\frac{\bar{n}_T(\upsilon)}{\omega+\upsilon}\right]\right\}.
\end{align}
%\end{widetext}

\twocolumngrid


\begin{thebibliography}{00}

%%%% GENERAL NON-MARKOVIAN DYNAMICS: DESCRIPTION AND CHARACTERIZATION %%%%

\bibitem{LindbladNoMarko} G. Lindblad, Commun. Math. Phys. \textbf{65}, 281 (1979).
%
\bibitem{Diosi} W. T. Strunz, L. Di\'osi, and N. Gisin, Phys. Rev. Lett. \textbf{82}, 1801 (1999).
%
\bibitem{BrPe02} H.-P. Breuer and F. Petruccione, \textit{The Theory of Open Quantum Systems} (Oxford University Press, Oxford, 2002).
%
\bibitem{GardinerZoller04} C. W. Gardiner and P. Zoller, \textit{Quantum Noise} (Springer, Berlin, 2004).
%
\bibitem{AlonsoVega} D. Alonso and I. de Vega, Phys. Rev. Lett. \textbf{94}, 200403 (2005).
%
\bibitem{Wolf} M. M. Wolf, J. Eisert, T. S. Cubitt, and J. I. Cirac, Phys. Rev. Lett. \textbf{101}, 150402 (2008).
%
\bibitem{BrLaPi} H.-P. Breuer, E.-M. Laine, and J. Piilo, Phys. Rev. Lett. \textbf{103}, 210401 (2009).
%
\bibitem{RHP} A. Rivas, S. F. Huelga, and M. B. Plenio, Phys. Rev. Lett. \textbf{105}, 050403 (2010).
%
\bibitem{Libro} A. Rivas and S. F. Huelga, \textit{Open Quantum Systems. An Introduction} (Springer, Heidelberg, 2011).
%
\bibitem{FlemingHu} C. H. Fleming and B. L. Hu, Ann. Phys. (N.Y.) \textbf{327}, 1238 (2012).
%
\bibitem{ChruMani2014} D. Chru\'sci\'nski and S. Maniscalco, Phys. Rev. Lett. \textbf{112}, 120404 (2014).
%
\bibitem{Review} A. Rivas, S.~F. Huelga, and M.~B. Plenio, Rep. Prog. Phys. \textbf{77}, 094001 (2014).
%
\bibitem{ReviewBLPV} H.-P. Breuer, E.-M. Laine, J. Piilo, and B. Vacchini, Rev. Mod. Phys. \textbf{88}, 021002 (2016).
%
\bibitem{InesReview} I. de Vega and D. Alonso, Dynamics of non-Markovian open quantum systems, e-print arXiv:1511.06994.
%
\bibitem{Kavan} F. A. Pollock, C. Rodr\'iguez-Rosario, T. Frauenheim, M. Paternostro, and K. Modi, Complete framework for efficient characterisation of non-Markovian processes, e-print arXiv:1512.00589.

%%%% POTENTIAL UTILITY OF NON-MARKOVIAN DYNAMICS%%%%

\bibitem{Sun} X.-M. Lu, X. Wang, and C. P. Sun, Phys. Rev. A \textbf{82}, 042103 (2010).

\bibitem{Chin-metrology} A. W. Chin, S. F. Huelga, and M. B. Plenio, Phys. Rev. Lett. \textbf{109}, 233601 (2012).

\bibitem{Matsuzaki} Y. Matsuzaki, S. C. Benjamin, and J. Fitzsimons, Phys. Rev. A \textbf{84}, 012103 (2011).

\bibitem{DAA} D. Chru\'sci\'nski, A. Kossakowski, and A. Rivas, Phys. Rev. A \textbf{83}, 052128 (2011).

\bibitem{Huelga} S. F. Huelga, A. Rivas, and M. B. Plenio, Phys. Rev. Lett. \textbf{108}, 160402 (2012).

\bibitem{Cialdi} S. Cialdi, D. Brivio, E. Tesio, and M.~G.~A. Paris, Phys. Rev. A \textbf{83}, 042308 (2011).

\bibitem{Xu} J.-S. Xu, K. Sun, C.-F. Li, X.-Y. Xu, G.-C. Guo, E. Andersson, R. Lo Franco, and G. Compagno, Nat. Commun. \textbf{4}, 2851 (2013).

\bibitem{Orieux} A. Orieux, G. Ferranti, A. D'Arrigo, R. Lo Franco, G. Benenti, E. Paladino, G. Falci, F. Sciarrino, and P. Mataloni, Sci. Rep. \textbf{5}, 8575 (2015).

\bibitem{Chen} H. Yang, H. Miao, and Y. Chen, Reveal non-Markovianity of open quantum systems via local operations, e-print arXiv:1111.6079.

\bibitem{Laine} E.-M. Laine, H.-P. Breuer, and J. Piilo, Sci. Rep. \textbf{4}, 4620 (2014).

\bibitem{Bylicka} B. Bylicka, D. Chru\'sci\'nski, and S. Maniscalco, Sci. Rep. \textbf{4}, 5720, (2014).

\bibitem{Xiang} G.-Y. Xiang, Z.-B. Hou, C.-F. Li, G.-C. Guo, H.-P. Breuer, E.-M. Laine, and J. Piilo, Europhys. Lett. \textbf{107}, 54006 (2014).

%%%% NONPOSITIVITY PROBLEM AT SHORT TIMES %%%%

\bibitem{RevKoss}
V. Gorini, A. Frigerio, M. Verri, A. Kossakowski, and E.~C.~G. Sudarshan, Rep. Math. Phys. \textbf{13}, 149 (1977).

\bibitem{Alicki} R. Alicki and K. Lendi \textit{Quantum Dynamical Semigroups and Applications} (Springer, Berlin, 1987).

\bibitem{Redfield} A. G. Redfield, IBM J. Res. Dev. \textbf{1}, 19 (1957).

\bibitem{Dumcke} R. D\"umcke and H. Spohn, Z. Phys. \textbf{B34}, 419 (1979).

\bibitem{Whitney} R. S. Whitney, J. Phys. A: Math. Theor. \textbf{41}, 175304 (2008).

\bibitem{Zhao} Y. Zhao and G. H. Chen, Phys. Rev. E \textbf{65}, 056120 (2002).

\bibitem{Suarez} A. Su\'arez, R. Silbey, and I. Oppenheim, J. Chem. Phys \textbf{97}, 5101 (1992).

\bibitem{Gaspard} P. Gaspard and M. Nagaoka, J. Chem. Phys. \textbf{111}, 5668 (1999).

\bibitem{Benatti2003} F. Benatti, R. Floreanini, and M. Piani, Phys. Rev. A \textbf{67}, 042110 (2003).

\bibitem{Benatti2006} F. Benatti, R. Floreanini, and S. Breteaux, Laser Phys. \textbf{16}, 1395 (2006).

\bibitem{Benatti2007} S. Anderloni, F. Benatti, and R. Floreanini, J. Phys. A: Math. Theor. \textbf{40}, 1625 (2007).

%%%% DYNCAMIAL COARSE GRAINING REFINED WEAK COUPLING %%%%

\bibitem{Schaller} G. Schaller and T. Brandes, Phys. Rev. A \textbf{78}, 022106, (2008).

\bibitem{SKB} G. Schaller, G. Kie{\ss}lich, and T. Brandes, Phys. Rev. B \textbf{80}, 245107 (2009).

\bibitem{Zedler} P. Zedler, G. Schaller, G. Kiesslich, C. Emary, and T. Brandes, Phys. Rev. B \textbf{80}, 045309 (2009).

\bibitem{higher} G. Schaller, P. Zedler, and T. Brandes, Phys. Rev. A \textbf{79}, 032110 (2009).

\bibitem{Benatti} F. Benatti, R. Floreanini, and U Marzolino, Phys. Rev. A \textbf{81}, 012105 (2010).

%%%% SPIN-BOSON MODEL %%%%

\bibitem{Leggett} A. J. Leggett, S. Chakravarty, A. T. Dorsey, M. P. A. Fisher, A. Garg, and W. Zwerger, Rev. Mod. Phys. \textbf{59}, 1 (1987).

\bibitem{Weiss} U. Weiss, \textit{Quantum Dissipative Systems} (World Scientific, Singapore, 2001).

\bibitem{Chakravarty} S. Chakravarty and J. Rudnick, Phys. Rev. Lett. \textbf{75}, 501 (1995).

\bibitem{Bulla} R. Bulla, N.-H. Tong, and M. Vojta, Phys. Rev. Lett. \textbf{91}, 170601 (2003).

\bibitem{Chin} A. W. Chin, J. Prior, S. F. Huelga, and M. B. Plenio, Phys. Rev. Lett. \textbf{107}, 160601 (2011).

\bibitem{Garg} A. Garg, J. N. Onuchic, and V. Ambegaokar, J. Chem. Phys. \textbf{83}, 4491 (1985).

\bibitem{Egger} L. Muhlbacher and R. Egger, Chem. Phys. \textbf{296}, 193 (2004).

\bibitem{Nitzan} A. Nitzan, \textit{Chemical Dynamics in Condensed Phases: Relaxation, Transfer, and Reactions
in Condensed Molecular Systems} (Oxford University Press, New York, 2006).

\bibitem{Cohen} C. Cohen-Tannoudji, J. Dupont-Roc, G. Grynberg, \textit{Atom-Photon Interactions} (Wiley, New York, 1992).

\bibitem{Diego} D. Porras, F. Marquardt, J. von Delft, and J. I. Cirac, Phys. Rev. A \textbf{78}, 010101(R) (2008).

\bibitem{Recati} A. Recati, P. O. Fedichev, W. Zwerger, J. von Delft, and P. Zoller, Phys. Rev. Lett. \textbf{94}, 040404 (2005).

\bibitem{Shnirman} A. Shnirman, Y. Makhlin, and G. Sch\"on, Phys. Scripta \textbf{T102}, 147 (2002).

\bibitem{Loss} D. P. DiVincenzo and D. Loss, Phys. Rev. B \textbf{71}, 035318 (2005).

\bibitem{Niemczyk} T. Niemczyk \textit{et al.}, Nat. Phys. \textbf{6}, 772 (2010).


%%%% NONMARKOVIANITY SPIN-BOSON MODEL %%%%

\bibitem{ClosBreuer2012}  In  ``G. Clos and H.-P. Breuer, Phys. Rev. A \textbf{86}, 012115 (2012)'' the non-Markovianity of the spin-boson model was analyzed by using the second order time-convolutionless master equation. The authors find traits of non-Markovianity in the behavior of the trace distance. However, since such an equation does not preserve complete positivity --in particular for long times where it converges to the Redfield equation, there is no guarantee that those non-Markovianity traits in the trace distance were due to the rough approximation or indeed correspond to the real behavior of the system.

%%%% GKSL form %%%%%%%%%%%

\bibitem{Koss-Lind}
V. Gorini, A. Kossakowski, and E.~C.~G. Sudarshan, J. Math. Phys. \textbf{17}, 821 (1976); G. Lindblad, Commun. Math. Phys. \textbf{48}, 119 (1976).


\bibitem{Davies}
E. B. Davies, Commun. Math. Phys. \textbf{39} 91 (1974).

%%%%%% Wilcox

\bibitem{dExpSM} R. F. Snider, J. Math. Phys. \textbf{5}, 1580 (1964); R. M. Wilcox, \textit{ibid.} \textbf{8}, 962 (1967).

%%%% Quantigiying coherence %%%%

\bibitem{QC}
T. Baumgratz, M. Cramer, and M. B. Plenio, Phys. Rev. Lett. \textbf{113}, 140401 (2014).

\bibitem{inco}
T. Chanda and S. Bhattacharya, Ann. Phys. \textbf{366}, 1 (2016).

%%%% Logarithmic negativity %%%%%
\bibitem{LogNeg} M.B. Plenio, Phys. Rev. Lett. \textbf{95}, 090503 (2005).

%%%%% Canonical decay rates %%%%
\bibitem{MichaelHall}
M. J. W. Hall, J. D. Cresser, L. Li, and E. Andersson, Phys. Rev. A \textbf{89}, 042120 (2014).

%%%% Markovian Master Equations
\bibitem{MME}
For the case of an Ohmic spectral density with exponential cutoff the behavior of the correlation functions is carefully studied in Sec. 2.3.2 of
A. Rivas, A. D. K. Plato, S. F. Huelga, and M. B. Plenio, New J. Phys. \textbf{12} 113032 (2010).

%%%%% Ferialdi

\bibitem{Ferialdi}
L. Ferialdi, Phys. Rev. A \textbf{95}, 020101(R) (2017). 

%%%% Boccara

\bibitem{BoccaraSM} See, for instance, N. Bocccara, \textit{Functional Analysis: An Introduction for Physicists} (Academic Press, San
Diego, 1990).


\end{thebibliography}
\end{document}